\documentclass[12pt]{amsart}

\usepackage{amsthm}
\usepackage{amssymb}
\usepackage{amsmath}

\numberwithin{equation}{section}

\newtheorem{thm}{Theorem}[section]
\newtheorem{prop}[thm]{Proposition}
\newtheorem{lem}[thm]{Lemma}
\newtheorem{cor}[thm]{Corollary}

\theoremstyle{definition}

\newtheorem{defn}[thm]{Definition}

\begin{document}

\title[Algebraic Construction of multi-species $q$-Boson system]
{Algebraic Construction of \\ multi-species $q$-Boson system}
\author{Yoshihiro Takeyama}
\address{Division of Mathematics, 
Faculty of Pure and Applied Sciences, 
University of Tsukuba, Tsukuba, Ibaraki 305-8571, Japan}
\email{takeyama@math.tsukuba.ac.jp}

\begin{abstract}
We construct a stochastic particle system which is a multi-species version of 
the $q$-Boson system due to Sasamoto and Wadati. 
Its transition rate matrix is obtained from a representation 
of a deformation of the affine Hecke algebra of type $GL$. 
\end{abstract}

\maketitle

\setcounter{section}{0}
\setcounter{equation}{0}


\section{Introduction}

In this article we construct a multi-species version of the $q$-Boson system due to 
Sasamoto and Wadati \cite{SW} by using a representation of a deformation of the affine Hecke algebra. 

The $q$-Boson system is a stochastic particle system on  
the one-dimensional lattice $\mathbb{Z}$. 
The particles can occupy the same site simultaneously, 
and one particle may move from site $i$ to $i-1$ independently for each $i \in \mathbb{Z}$. 
The rate at which one particle moves from a cluster with $n$ particles is given by 
$1-q^{n}$, where $q$ is a parameter of the model. 

The multi-species version which we propose in this paper is described as follows. 
Fix a positive integer $N$. 
Each bosonic particle is colored with a positive integer which is less than or equal to $N$. 
One particle may move to the left in the same way as the $q$-Boson system, 
but the rate is different. 
Let $b \in \{1, 2, \ldots, N\}$ be the color of the moving particle 
and $m_{j} \, (j=b, b+1, \ldots , N)$ the number of particles with color $j$ in the cluster from which 
the moving particle leaves. 
Then the rate is given by 
\begin{align*}
\frac{1-q^{m_{b}}}{1-q}q^{\sum_{j=b+1}^{N}m_{j}}.
\end{align*}
If $N=1$, the transition rate matrix is equal to that of the $q$-Boson system 
up to constant multiplication. 

In the following we describe how the multi-species model arises from 
representation theory. 
Hereafter we fix a positive integer $k$ which signifies the number of particles. 

In a previous paper \cite{T}, 
we introduced a deformation of the affine Hecke algebra of type $GL_{k}$ with four parameters 
and found that an integrable stochastic particle system can be constructed from 
its representation as follows. 
The deformed algebra has a representation on the vector space $F(L)$ 
of $\mathbb{C}$-valued functions on the $k$-dimensional orthogonal lattice $L:=\mathbb{Z}^{k}$. 
It is defined in terms of a generalization of 
the discrete integral-reflection operators due to van Diejen and Emsiz \cite{DE}. 
Using them we define an operator $G: F(L) \to F(L)$, 
which is a discrete analogue of the propagation operator introduced by Gutkin \cite{G} 
in order to construct eigenfunctions of the Hamiltonian with delta potential, 
and determine the discrete Hamiltonian $H$ by the property 
$HG=G\sum_{i=1}^{k}t_{i}$, where $t_{i}$ is the shift operator 
$(t_{i}f)(x_{1}, \ldots , x_{k}):=f(\ldots , x_{i}-1, \ldots)$. 
Then the operator $H$ leaves the subspace $F(L)^{\mathfrak{S}_{k}}$ 
of symmetric functions invariant. 
Specializing the parameters and 
restricting $H$ to $F(L)^{\mathfrak{S}_{k}}$, 
we obtain the transition rate matrix of a continuous-time Markov process.  
The resulting system is a continuous-time limit of 
the $q$-Hahn system due to Povolotsky \cite{P, BC}.  

In this paper we generalize the above construction in a similar way 
to the generalization of the periodic delta Bose gas due to 
Emsiz, Opdam and Stokman \cite{EOS}. 
In this article we only consider the case where one of the parameters of the deformed algebra is equal to zero. 
Then the algebra, which we denote by $\mathcal{A}_{k}$, 
essentially has two parameters $\alpha$ and $q$. 
The algebra $\mathcal{A}_{k}$ contains the Hecke algebra $\mathcal{H}_{k}$ of type $A_{k-1}$ as a subalgebra. 
Let $M$ be a left $\mathcal{H}_{k}$-module and 
denote by $F(L, M)$ the vector space of functions on $L$ taking values in $M$. 
Then we can define an action of $\mathcal{A}_{k}$ on $F(L, M)$, 
introduce the propagation operator $G$ and 
determine the discrete Hamiltonian $H$ acting on $F(L, M)$ from the property 
$HG=G\sum_{i=1}^{k}t_{i}$.  
Then the Hamiltonian $H$ leaves a subspace $F_{0}(L, M)$ 
(see Definition \ref{defn:invariant-subspace} below) of $F(L, M)$ invariant. 

Now the multi-species version of the $q$-Boson system is constructed as follows. 
Let $U$ be an $N$-dimensional vector space. 
We regard $U^{\otimes k}$ as a left $\mathcal{H}_{k}$-module with respect to 
the action found by Jimbo \cite{J}. 
Then the invariant subspace $F_{0}(L, U^{\otimes k})$ is identified with 
the vector space of functions $F(\mathcal{S})$ on the set of configurations of 
$k$ bosonic particles of $N$ species on the one-dimensional lattice $\mathbb{Z}$. 
Setting $\alpha=-(1-q)$ and restricting the Hamiltonian $H$ to $F(\mathcal{S})$, 
we obtain the transition rate matrix of the multi-species $q$-Boson system 
up to an additive constant. 

Using the propagation operator $G: F(L, M)\to F(L, M)$, 
we can construct eigenfunctions of the discrete Hamiltonian $H$ 
by means of the Bethe ansatz method, which we call the Bethe wave functions. 
We should construct Plancherel theory for them to analyze the multi-species $q$-Boson system 
in a similar manner to 
the work of Borodin, Corwin, Petrov and Sasamoto \cite{BCPS1, BCPS2}. 
We leave it as a future problem. 

The paper is organized as follows. 
In Section \ref{sec:hecke} we introduce the deformation of 
the affine Hecke algebra and its representation 
defined by the discrete integral-reflection operators. 
In Section \ref{sec:Hamiltonian} we define the propagation operator 
and the discrete Hamiltonian, 
and construct the Bethe wave functions. 
In Section \ref{sec:system} we derive the transition rate matrix 
of the multi-species $q$-Boson system from the discrete Hamiltonian.


\section{A deformation of the affine Hecke algebra and its representation}\label{sec:hecke}

\subsection{Preliminaries}

Throughout this paper we fix an integer $k \ge 2$. 
Let $V:=\oplus_{i=1}^{k}\mathbb{R}v_{i}$ be the $k$-dimensional Euclidean space  
with the orthonormal basis $\{v_{i}\}_{i=1}^{k}$. 
We denote by $V^{*}$ the linear dual of $V$ 
and by $\{\epsilon_{i}\}_{i=1}^{k}$ the dual basis of $V^{*}$ corresponding to $\{v_{i}\}_{i=1}^{k}$.  
For $i, j=1, \ldots , k$ we set $\alpha_{ij}:=\epsilon_{i}-\epsilon_{j}$. 
Then the set $R:=\{\alpha_{ij} \, | \, i\not=j\}$ forms 
the root system of type $A_{k-1}$ with the simple roots $a_{i}:=\alpha_{i, i+1} \, (1\le i<k)$. 
We denote the set of the associated positive and negative roots by $R^{+}$ and $R^{-}$, respectively.  

Let $s_{i} \, : \, V \to V \, (1\le i<k)$ be the orthogonal reflection
\begin{align*}
s_{i}(v):=v-a_{i}(v)a^{\vee}_{i}, 
\end{align*}
where $a^{\vee}_{i}:=v_{i}-v_{i+1}$ is the simple coroot.  
The Weyl group $W$ of type $A_{k-1}$ is generated by the simple reflections $\{s_{i}\}_{i=1}^{k-1}$. 
We denote the length of $w \in W$ by $\ell(w)$. 
The dual space $V^{*}$ is a $W$-module by 
$(w\lambda)(v):=\lambda(w^{-1}v) \, (w \in W, \lambda \in V^{*}, v \in V)$. 

Let $v \in V$. 
The orbit $Wv$ intersects the closure of the fundamental chamber 
\begin{align*}
\overline{C_{+}}:=\{v \in V \, | \, a_{i}(v)\ge 0 \,\, (i=1, \ldots , k-1)\}  
\end{align*}
at one point. 
Take the shortest element $w \in W$ such that $wv \in \overline{C_{+}}$. 
We denote it by $w_{v}$. 
Set 
\begin{align*}
I(v):=\left\{ a \in R^{+} \, | \, a(v)<0 \right\}.  
\end{align*}
If $w_{v}=s_{i_{1}}\cdots s_{i_{\ell}}$ is a reduced expression, 
then $I(v)=\{ s_{i_{\ell}}\cdots s_{i_{p+1}}(a_{i_{p}})\}_{p=1}^{\ell}$. 
Therefore $\ell(w_{v})=\# I(v)$ and $I(v)=R^{+}\cap w_{v}^{-1}R^{-}$. 
{}From the above facts, we obtain the following lemma. 

\begin{lem}\label{lem:I}
Suppose that $I(v_{1}) \subset I(v_{2})$. Then 
$w_{v_{2}}=w_{w_{v_{1}}v_{2}}w_{v_{1}}$ and $\ell(w_{v_{2}})=\ell(w_{w_{v_{1}}v_{2}})+\ell(w_{v_{1}})$.  
\end{lem}

We denote by $L$ the $k$-dimensional orthogonal lattice in $V$ defined by 
\begin{align*}
L=\bigoplus_{i=1}^{k}\mathbb{Z}v_{i},   
\end{align*}
and set 
\begin{align*}
L_{+}:=L \cap \overline{C_{+}}.  
\end{align*}
Hereafter we set $\epsilon_{k+1}(x)=-\infty$ for $x \in L_{+}$.

For $x \in L$ and $1\le i \le k$, we set 
\begin{align*}
d_{i}^{+}(x)&:=\#\left\{ p \, | \, i<p \le k, \, \epsilon_{i}(x)=\epsilon_{p}(x) \right\}, \\ 
d_{i}^{-}(x)&:=\#\left\{ p \, | \, 1\le i<p, \, \epsilon_{i}(x)=\epsilon_{p}(x) \right\}. 
\end{align*}
We denote by $\sigma_{x} \in \mathfrak{S}_{k}$ the permutation 
determined by $w_{x}v_{i}=v_{\sigma_{x}(i)} \, (1\le i \le k)$.  
Then it holds that  
\begin{align}
d_{i}^{\pm}(x)=d_{\sigma_{x}(i)}^{\pm}(w_{x}x), \quad 
\epsilon_{i}(x)=\epsilon_{\sigma_{x}(i)}(w_{x}x) 
\label{eq:d-change}
\end{align}
for $x \in L$ and $1\le i\le k$.

\begin{prop}\label{prop:shortest-prod}
Suppose that $1 \le i<k$. 
\begin{enumerate}
 \item If $v \in V$ satisfies $a_{i}(v)>0$, then 
$w_{s_{i}v}=w_{v}s_{i}$ and $\ell(w_{s_{i}v})=\ell(w_{v})+1$. 
 \item If $x \in L$ satisfies $a_{i}(x)=0$, then 
$\sigma_{x}(i+1)=\sigma_{x}(i)+1, w_{x}s_{i}=s_{\sigma_{x}(i)}w_{x}$ and 
$\ell(w_{x}s_{i})=\ell(s_{\sigma_{x}(i)}w_{x})=\ell(w_{x})+1$. 
\end{enumerate}
\end{prop}

\begin{proof}
(i)\, From $a_{i}(v)>0$ we see that $I(s_{i}v)=s_{i}I(v) \sqcup \{a_{i}\}$. 
Hence $\ell(w_{s_{i}v})=\ell(w_{v})+1$. 
Since $w_{v}s_{i}$ moves $s_{i}v$ into $\overline{C_{+}}$, 
we have $w_{s_{i}v}=w_{v}s_{i}$. 

(ii)\, 
Since $a_{i}(x)=0$ and $w_{x}x \in L_{+}$, 
there exists an integer $p$ such that 
$1\le p \le k$ and 
$\epsilon_{\sigma_{x}(i)}(w_{x}x)=\epsilon_{\sigma_{x}(i+1)}(w_{x}x)=\epsilon_{p}(w_{x}x)>\epsilon_{p+1}(w_{x}x)$.  
Then, for $j=i$ and $i+1$, 
it holds that 
$p-\sigma_{x}(j)=d_{\sigma_{x}(j)}^{+}(w_{x}x)=d_{j}^{+}(x)$.
On the other hand, $d_{i}^{+}(x)=d_{i+1}^{+}(x)+1$ because $a_{i}(x)=0$. 
Therefore 
\begin{align*}
\sigma_{x}(i+1)=p-d_{i+1}^{+}(x)=p-d_{i}^{+}(x)+1=\sigma_{x}(i)+1.  
\end{align*}

Set $z:=x-v_{i+1}/2$. 
Since $|a(v_{j}/2)|\le 1/2 \, (j=i, i+1)$ for any $a \in R$, 
we have $I(x) \subset I(z)$ and $I(x) \subset I(s_{i}z)$. 
Moreover, $w_{s_{i}z}=w_{z}s_{i}$ and $\ell(w_{s_{i}z})=\ell(w_{z})+1$ from (i). 
Therefore, using Lemma \ref{lem:I}, we get 
\begin{align}
w_{w_{x}z}w_{x}s_{i}=w_{w_{x}s_{i}z}w_{x}, \quad   
\ell(w_{w_{x}z})+\ell(w_{x})+1=\ell(w_{w_{x}s_{i}z})+\ell(w_{x}). 
\label{eq:shortest-prod}
\end{align}
Note that $w_{x}z=w_{x}x-v_{\sigma_{x}(i+1)}/2$ and $w_{x}s_{i}z=w_{x}x-v_{\sigma_{x}(i)}/2$. 
Since $w_{x}x \in L_{+}$, we have 
\begin{align*}
I(w_{x}z)=\{ \alpha_{\sigma_{x}(i+1), l} \}_{l=\sigma_{x}(i+1)+1}^{p}, \quad 
I(w_{x}s_{i}z)=\{ \alpha_{\sigma_{x}(i), l} \}_{l=\sigma_{x}(i)+1}^{p}. 
\end{align*}
Hence 
\begin{align*}
w_{w_{x}z}=s_{p-1}s_{p-2}\cdots s_{\sigma_{x}(i+1)}, \quad 
w_{w_{x}s_{i}z}=s_{p-1}s_{p-2}\cdots s_{\sigma_{x}(i)},  
\end{align*}
where the right hand sides are reduced expressions. 
Using $\sigma_{x}(i+1)=\sigma_{x}(i)+1$ and \eqref{eq:shortest-prod}, 
we find that $w_{x}s_{i}=s_{\sigma_{x}(i)}w_{x}$ and 
$\ell(w_{x}s_{i})=\ell(s_{\sigma_{x}(i)}w_{x})=\ell(w_{x})+1$. 
\end{proof}

We will also use the following proposition. 
See Lemma 3.9 and Lemma 3.10 in \cite{T} for the proof.  

\begin{prop}\label{prop:shortest-shift}
Suppose that $x \in L$ and $1 \le i \le k$. 
Set $y=x-v_{i}$. Then it holds that 
\begin{align}
& 
s_{\sigma_{y}(i)-d_{i}^{-}(y)} \cdots s_{\sigma_{y}(i)-1}w_{y}=
s_{\sigma_{x}(i)+d_{i}^{+}(x)-1} \cdots s_{\sigma_{x}(i)}w_{x},  \nonumber \\ 
& 
a_{j}(w_{x}x)=0 \quad (\sigma_{x}(i)\le j<\sigma_{x}(i)+d_{i}^{+}(x)), 
\label{eq:main} \\ 
& 
a_{j}(w_{y}y)=0 \quad (\sigma_{y}(i)-d_{i}^{-}(y)\le j<\sigma_{y}(i)) \nonumber 
\end{align} 
and $d_{i}^{-}(y)+\ell(w_{y})=d_{i}^{+}(x)+\ell(w_{x})$. 
\end{prop}

\subsection{A deformation of the affine Hecke algebra and its representation}

\begin{defn}
Let $\alpha$ and $q$ be complex constants.  
We define the algebra $\mathcal{A}_{k}$ to be the unital associative $\mathbb{C}$-algebra with 
the generators $X_{i}^{\pm 1} \, (1\le i \le k)$ and $T_{i} \, (1 \le i<k)$ satisfying 
the following relations: 
\begin{align*}
& 
(T_{i}-1)(T_{i}+q)=0 \quad (1 \le i<k), \qquad  
T_{i}T_{i+1}T_{i}=T_{i}T_{i+1}T_{i} \quad (1 \le i \le k-2), \\ 
& 
T_{i}T_{j}=T_{j}T_{i} \quad (|i-j|>1), \quad 
X_{i}X_{j}=X_{j}X_{i} \quad (i, j=1, \ldots , k), \\ 
& 
X_{i+1}T_{i}-T_{i}X_{i}=T_{i}X_{i+1}-X_{i}T_{i}=(1-q)X_{i+1}+\alpha \quad 
(1\le i<k), \\ 
& 
X_{i}T_{j}=T_{j}X_{i}\quad (i\not=j, j+1).  
\end{align*} 
\end{defn}

When $\alpha=0$, the algebra $\mathcal{A}_{k}$ is isomorphic to the affine Hecke algebra of type $GL_{k}$. 
The subalgebra generated by $T_{i} \, (1 \le i <k)$ is isomorphic to 
the Hecke algebra of type $A_{k-1}$.  
We denote it by $\mathcal{H}_{k}$.  

For a left $\mathcal{H}_{k}$-module $M$, we denote by $F(L, M)$ the complex vector space of 
functions on $L$ taking values in $M$. 
The Weyl group $W$ acts on $F(L, M)$ by 
$(wf)(x):=f(w^{-1}x) \, (w \in W, f \in F(L, M), x \in L)$. 
Let $\widehat{T}_{i} \, (1 \le i<k)$ be the $\mathbb{C}$-linear operator on $F(L, M)$ 
defined by $(\widehat{T}_{i}f)(x)=T_{i}.f(x) \, (f \in F(L, M), x \in L)$, 
where $.$ signifies the action of $\mathcal{H}_{k}$ on $M$. 
Then the assignment $T_{i} \mapsto \widehat{T}_{i} \, (1 \le i<k)$ uniquely determines 
an algebra homomorphism $\mathcal{H}_{k} \to \mathrm{End}(F(L, M))$. 
It commutes with the action of $W$. 

We identify the group algebra $\mathbb{C}[L]$ with the Laurent polynomial ring 
$\mathbb{C}[e^{\pm v_{1}}, \ldots , e^{\pm v_{k}}]$. 
The Weyl group acts on $\mathbb{C}[L]$ from the right by 
$e^{x}w=e^{w^{-1}x} \,(x \in L, w \in W)$. 
Using this action we define the $\mathbb{C}$-linear map 
$\check{I}_{j} : \mathbb{C}[L] \to \mathbb{C}[L] \, (1 \le j <k)$ by 
\begin{align*}
\check{I}_{j}(P):=(P-Ps_{j})\frac{\alpha e^{v_{j+1}}+1-q}{1-e^{-v_{j}+v_{j+1}}}.  
\end{align*}

Consider the non-degenerate $\mathbb{C}$-bilinear pairing 
$( \quad , \quad): \mathbb{C}[L] \times F(L, M) \to M$ uniquely determined 
by $(e^{x}, f)=f(x) \, (x \in L, f \in F(L, M))$. 
Now we define the $\mathbb{C}$-linear operator 
$\widehat{I}_{j}: F(L, M) \to F(L, M) \, (1 \le j<k)$ by the property 
\begin{align*}
(P, \widehat{I}_{j}(f))=(\check{I}_{j}(P), f)
\end{align*}
for any $P \in \mathbb{C}[L]$. 
It is explicitly written as follows. 
\begin{align}
(\widehat{I}_{j}f)(x)=\left\{ 
\begin{array}{ll}
\displaystyle 
\sum_{l=0}^{a_{j}(x)-1}\left( 
\alpha f(x-l a_{j}^{\vee}+v_{j+1})+(1-q)f(x-l a_{j}^{\vee})
\right) & (a_{j}(x)>0) \\
0 & (a_{j}(x)=0) \\
\displaystyle  
{}-\sum_{l=1}^{-a_{j}(x)}\left( 
\alpha f(x+la_{j}^{\vee}+v_{j+1})+(1-q)f(x+la_{j}^{\vee})
\right) & (a_{j}(x)<0) 
\end{array}
\right.
\label{eq:integration-op}
\end{align} 

\begin{lem}\label{lem:comm-rel}
The following relations hold in $\mathrm{End}(F(L, M))$. 
\begin{enumerate}
 \item $\widehat{I}_{j}^{\, 2}=(1-q)\widehat{I}_{j}$ for $1\le j <k$. 
 \item $s_{i}\widehat{I}_{j}=\widehat{I}_{j}s_{i}$ if $|i-j|\ge 2$. 
 \item $s_{j}\widehat{I}_{j}+\widehat{I}_{j}s_{j}=(1-q)(s_{j}-1)$ for $1\le j<k$. 
 \item $\widehat{I}_{j}s_{j+1}s_{j}=s_{j+1}s_{j}\widehat{I}_{j+1}$ 
and $\widehat{I}_{j+1}s_{j}s_{j+1}=s_{j}s_{j+1}\widehat{I}_{j}$  
for $1 \le j \le k-2$. 
 \item For $1 \le j \le k-2$, 
\begin{align*}
\widehat{I}_{j}s_{j+1}\widehat{I}_{j}&=
(1-q)s_{j+1}\widehat{I}_{j}s_{j+1}+s_{j+1}\widehat{I}_{j}\widehat{I}_{j+1}+
\widehat{I}_{j+1}\widehat{I}_{j}s_{j+1}, \\
\widehat{I}_{j+1}s_{j}\widehat{I}_{j+1}&=
(1-q)s_{j}\widehat{I}_{j+1}s_{j}+s_{j}\widehat{I}_{j+1}\widehat{I}_{j}+
\widehat{I}_{j}\widehat{I}_{j+1}s_{j}. 
\end{align*}
 \item $\widehat{I}_{j}\widehat{I}_{j+1}\widehat{I}_{j}+qs_{j}\widehat{I}_{j+1}s_{j}=
\widehat{I}_{j+1}\widehat{I}_{j}\widehat{I}_{j+1}+qs_{j+1}\widehat{I}_{j}s_{j+1}$ 
for $1\le j<k$.   
\end{enumerate}
\end{lem}

\begin{proof}
Straightforward check.  
\end{proof}

\begin{prop}
Let $M$ be a left $\mathcal{H}_{k}$-module. 
For $1 \le j \le k$, let $t_{j}: F(L, M) \to F(L, M)$ be the shift operator 
\begin{align*}
(t_{j}f)(x):=f(x-v_{j}) \quad (f \in F(L, M), x \in L). 
\end{align*}
Then the assignments 
\begin{align*}
X_{j} \mapsto t_{j} \quad (1\le j\le k), \qquad 
T_{j} \mapsto \widehat{T}_{j}s_{j}+\widehat{I}_{j} \quad (1 \le j<k)  
\end{align*}
uniquely extend to a $\mathbb{C}$-algebra homomorphism 
$\rho: \mathcal{A}_{k} \to \mathrm{End}(F(L, M))$. 
\end{prop}

\begin{proof}
Note that the operators $\widehat{I}_{j} \, (1\le j<k)$ commute 
with $\widehat{T}_{j}\, (1\le j<k)$, 
and 
\begin{align*}
(e^{-v_{j}}P, f)=(P, t_{j}f).  
\end{align*} 
for $P \in \mathbb{C}[L]$ and $f \in F(L, M)$.  
Now we can check the defining relations of $\mathcal{A}_{k}$ by using Lemma \ref{lem:comm-rel} and 
the equality
\begin{align*}
\check{I}_{j}(e^{x-v_{j+1}})-e^{-v_{j}}\check{I}_{j}(e^{x})=((1-q)e^{-v_{j+1}}+\alpha)e^{x}=
\check{I}_{j}(e^{x-v_{j}})-e^{-v_{j+1}}\check{I}_{j}(e^{x})
\end{align*}
for $1\le j<k$ and $x \in L$. 
\end{proof}

We will often use the fact below which follows from \eqref{eq:integration-op}.  

\begin{prop}\label{prop:action-trivial}
Let $M$ be a left $\mathcal{H}_{k}$-module. 
Suppose that $1\le j<k$, $x \in L$ and $f \in F(L, M)$. 
If $a_{j}(x)=0$, then $(\rho(T_{j})f)(x)=T_{j}.f(x)$. 
\end{prop}


\section{Discrete Hamiltonian}\label{sec:Hamiltonian}

\subsection{Discrete Hamiltonian and Propagation operator}

For $1 \le i \le k$ and $x \in L$, 
we define $T_{i}^{(\pm)}(x) \in \mathcal{H}_{k}$ by 
\begin{align*}
T_{i}^{(-)}(x)&:=T_{w_{x}}^{-1}\left(
T_{\sigma_{x}(i)-1}^{-1}\cdots T_{\sigma_{x}(i)-d_{i}^{-}(x)}^{-1}
\right) 
\left(
T_{\sigma_{x}(i)-d_{i}^{-}(x)}^{-1} \cdots T_{\sigma_{x}(i)-1}^{-1}
\right) 
T_{w_{x}}, \\ 
T_{i}^{(+)}(x)&:=T_{w_{x}}^{-1}\left(
\sum_{j=\sigma_{x}(i)}^{\sigma_{x}(i)+d_{i}^{+}(x)-1}
\left(T_{\sigma_{x}(i)}^{-1} \cdots T_{j-1}^{-1} \right)T_{j}^{-1}
\left(T_{j-1} \cdots T_{\sigma_{x}(i)} \right)
\right) T_{w_{x}}. 
\end{align*}

\begin{defn}
Let $M$ be a left $\mathcal{H}_{k}$-module. 
We define the discrete Hamiltonian $H : F(L, M) \to F(L, M)$ by 
\begin{align*}
(Hf)(x):=\sum_{i=1}^{k}q^{d_{i}^{-}(x)}T_{i}^{(-)}(x).\left( 
f(x-v_{i})-\alpha \, T_{i}^{(+)}(x).f(x)
\right) \quad (f \in F(L, M), x \in L),  
\end{align*}
where $.$ means the left action of $\mathcal{H}_{k}$ on $M$. 
\end{defn} 

Next we define the propagation operator. 
Let $w=s_{i_{1}} \cdots s_{i_{l}}$ be a reduced expression of $w \in W$.   
Then we set $T_{w}:=T_{i_{1}}\cdots T_{i_{l}}$. 

\begin{defn}
Let $M$ be a left $\mathcal{H}_{k}$-module. 
We define the \textit{propagation operator} $G: F(L, M) \to F(L, M)$ by 
\begin{align*}
(Gf)(x):=T_{w_{x}}^{-1}.\left((\rho(T_{w_{x}})f)(w_{x}x)\right) 
\quad (f \in F(L, M), x \in L), 
\end{align*} 
where $.$ means the action of $\mathcal{H}_{k}$ on $M$. 
\end{defn}

\begin{thm}
It holds that $HG=G(\sum_{i=1}^{k}t_{i})$. 
Therefore if $f$ is an eigenfunction of $\sum_{i=1}^{k}t_{i}$, then 
$G(f)$ is that of the discrete Hamiltonian $H$ with the same eigenvalue.  
\end{thm}

\begin{proof}
Suppose that $f \in F(L, M)$ and $x \in L$. 
First we fix $i \, (1 \le i \le k)$ and calculate $(Gf)(x-v_{i})$. 
Set $y=x-v_{i}$. From Proposition \ref{prop:shortest-shift}, we have 
\begin{align}
T_{w_{y}}^{-1}= 
T_{w_{x}}^{-1} (T_{\sigma_{x}(i)}^{-1} \cdots T_{\sigma_{x}(i)+d_{i}^{+}(x)-1}^{-1})
(T_{\sigma_{y}(i)-d_{i}^{-}(y)} \cdots T_{\sigma_{y}(i)-1}). 
\label{eq:main-0}
\end{align}
Using Proposition \ref{prop:action-trivial}, we see that 
\begin{align*}
(Gf)(x-v_{i})=T_{w_{x}}^{-1} (T_{\sigma_{x}(i)}^{-1} \cdots T_{\sigma_{x}(i)+d_{i}^{+}(x)-1}^{-1}). 
\left( \rho(T_{\sigma_{x}(i)+d_{i}^{+}(x)-1} \cdots T_{\sigma_{x}(i)})f
\right)\!(w_{y}y).  
\end{align*}
{}From \eqref{eq:main} it holds that $w_{y}y=w_{x}x-v_{\sigma_{x}(i)+d_{i}^{+}(x)}$. 
Therefore 
\begin{align*}
&
(Gf)(x-v_{i}) \\ 
&=T_{w_{x}}^{-1} (T_{\sigma_{x}(i)}^{-1} \cdots T_{\sigma_{x}(i)+d_{i}^{+}(x)-1}^{-1}). 
\left( \rho(X_{\sigma_{x}(i)+d_{i}^{+}(x)}T_{\sigma_{x}(i)+d_{i}^{+}(x)-1} \cdots T_{\sigma_{x}(i)})f
\right)\!(w_{x}x).  
\end{align*}
Move $X_{\sigma_{x}(i)+d_{i}^{+}(x)}$ to the right using the relation 
$X_{j+1}T_{j}=T_{j}X_{j}+\alpha+(1-q)X_{j+1}$, and 
use Proposition \ref{prop:action-trivial} to remove 
$\rho(T_{j}) \, (\sigma_{x}(i)\le j <\sigma_{x}(i)+d_{i}^{+}(x))$. 
As a result we get 
\begin{align}
& 
(Gf)(x-v_{i})=
T_{w_{x}}^{-1}.(\rho(X_{\sigma_{x}(i)}T_{w_{x}})f)(w_{x}x)+\alpha \, T_{i}^{(+)}(x).G(f)(x) 
\label{eq:main-05} \\ 
&+(1-q)\sum_{j=\sigma_{x}(i)}^{\sigma_{x}(i)+d_{i}^{+}(x)-1}
T_{w_{x}}^{-1}(T_{\sigma_{x}(i)}^{-1}\cdots T_{j-1}^{-1})T_{j}^{-1}(T_{j-1}\cdots T_{\sigma_{x}(i)}).
(\rho(X_{j+1}T_{w_{x}})f)(w_{x}x).  
\nonumber 
\end{align}

Now let us calculate $(HG(f))(x)$. 
The set $\{1, 2, \ldots , k\}$ is decomposed into a direct sum of sub-intervals 
$J_{1}, \ldots , J_{r}$ so that $i \in J_{a}$ and $j \in J_{a}$ if and only if 
$\epsilon_{i}(w_{x}x)=\epsilon_{j}(w_{x}x)$ for $a=1, 2, \ldots , r$. 
For an interval $J=\{l, l+1, \ldots , m\}$, set 
\begin{align*}
K_{J}&:=\sum_{i=l}^{m}q^{i-l}
(T_{i-1}^{-1}\cdots T_{l}^{-1})(T_{l}^{-1} \cdots T_{i-1}^{-1}) \\ 
& \quad {}\times 
\left\{X_{i}+(1-q)\sum_{j=i}^{m-1}(T_{i}^{-1}\cdots T_{j-1}^{-1})T_{j}^{-1}(T_{j-1}\cdots T_{i})X_{j+1}\right\}
\end{align*}
Change the index $i$ in the definition of $H$ to $\sigma_{x}(i)$ and 
rewrite $(HG(f))(x)$ using \eqref{eq:d-change} and \eqref{eq:main-05}.   
Then we obtain 
\begin{align*}
(HG(f))(x)=\sum_{a=1}^{r}T_{w_{x}}^{-1}.\left(\rho(K_{J_{a}}T_{w_{x}})f\right)\!(w_{x}x).  
\end{align*}
On the other hand, $K_{J}$ is rewritten as 
\begin{align*}
K_{J}=\sum_{i=l}^{m}&\bigg\{
q^{i-l}(T_{i-1}^{-1}\cdots T_{l}^{-1})(T_{l}^{-1}\cdots T_{i-1}^{-1}) \\ 
&+(1-q)\sum_{j=l}^{i-1}q^{j-l}
(T_{i-1}\cdots T_{j+1})(T_{j-1}^{-1}\cdots T_{l}^{-1})
(T_{l}^{-1} \cdots T_{i-1}^{-1})
\bigg\}X_{i}.  
\end{align*}
{}From the above expression, we see that $K_{J}=\sum_{i=l}^{m}X_{i}$ 
using the relation $qT_{j}^{-1}+(1-q)=T_{j}$. 
Thus we get 
\begin{align*}
(HG(f))(x)=T_{w_{x}}^{-1}.\left(\rho({\textstyle \sum_{i=1}^{k}X_{i}}T_{w_{x}})f\right)\!(w_{x}x).   
\end{align*}
Since $\sum_{i=1}^{k}X_{i}$ commutes with $T_{i} \, (1\le i<k)$, 
the right hand side is equal to 
\begin{align*}
T_{w_{x}}^{-1}.\left(\rho(T_{w_{x}})\rho({\textstyle \sum_{i=1}^{k}X_{i}})f\right)\!(w_{x}x)
=(G({\textstyle \sum_{i=1}^{k}t_{i}})f)(x).    
\end{align*}
This completes the proof. 
\end{proof}

\subsection{Invariant subspace}

\begin{defn}\label{defn:invariant-subspace}
Let $M$ be a left $\mathcal{H}_{k}$-module.  
We denote by $F_{0}(L, M)$ the subspace of $F(L, M)$ 
consisting of the functions $f: L \to M$ satisfying 
\begin{align*}
f(s_{i}x)=\left\{ 
\begin{array}{ll}
T_{i}^{-1}.f(x) &  (a_{i}(x)\ge 0) \\
T_{i}.f(x) & (a_{i}(x)<0)
\end{array}
\right.
\end{align*}
for $1\le i<k$ and $x \in L$. 
\end{defn}

Note that if $f \in F_{0}(L, M)$ then 
\begin{align}
f(x)=T_{w_{x}}^{-1}.f(w_{x}x) 
\label{eq:symmetric-value}
\end{align} 
for $x \in L$. 

In this subsection we prove the following theorem. 

\begin{thm}\label{thm:invariance}
Let $M$ be a left $\mathcal{H}_{k}$-module. 
Then it holds that $H(F_{0}(L, M)) \subset F_{0}(L, M)$. 
\end{thm}

For that purpose, we rewrite the formula of $Hf \, (f \in F_{0}(L, M))$. 
For $x \in L_{+}$, we define the {\it cluster coordinate} $(c_{1}, c_{2}, \ldots , c_{r})$ of $x$ 
by the property that $\sum_{a=1}^{r}c_{a}=k$, 
$\epsilon_{c_{1}}(x)>\epsilon_{c_{1}+c_{2}}(x)>\cdots >\epsilon_{c_{1}+\cdots +c_{r}}(x)$ and 
$\epsilon_{j}(x)=\epsilon_{c_{1}+\cdots +c_{a}}(x)$ if 
$c_{1}+\cdots +c_{a-1}<j\le c_{1}+\cdots +c_{a}$. 
For example, if $k=5$ and $x=2v_{1}+2v_{2}-v_{3}-3v_{4}-3v_{5}$, then 
the cluster coordinate of $x$ is $(2, 1, 2)$. 

\begin{prop}\label{prop:Hamiltonian-on-symmetric-space}
Let $M$ be a left $\mathcal{H}_{k}$-module. 
Suppose that $f \in F_{0}(L, M)$ and $x \in L$. 
Let $(c_{1}, \ldots , c_{r})$ be the cluster coordinate of $w_{x}x$. 
Then it holds that 
\begin{align}
(Hf)(x)&=T_{w_{x}}^{-1}
\sum_{a=1}^{r}\sum_{l=1}^{c_{a}}q^{c_{a}-l}
(T_{c_{1}+\cdots +c_{a-1}+l}^{-1} \cdots T_{c_{1}+\cdots +c_{a}-1}^{-1}).
f(w_{x}x-v_{c_{1}+\cdots +c_{a}}) 
\label{eq:Hamiltonian-on-symmetric-space-formula} \\ 
&-\frac{\alpha}{1-q} \sum_{a=1}^{r}(c_{a}-[c_{a}]_{q})f(x), 
\nonumber  
\end{align} 
where $[n]_{q}:=(1-q^{n})/(1-q)$ is the $q$-integer. 
\end{prop}

In the proof of Proposition \ref{prop:Hamiltonian-on-symmetric-space}, 
we use the following formula.  

\begin{lem}\label{lem:Hamiltonian-on-symmetric-space} 
Suppose that $f \in F_{0}(L, M), x \in L_{+}$ and $0\le p<p+c \le k$. 
If $a_{j}(x)=0$ for $p+1 \le j<p+c$, then it holds that 
\begin{align}
&
\sum_{l=1}^{c}q^{l-1}
(T_{p+l-1}^{-1} \cdots T_{p+2}^{-1}T_{p+1}^{-1}).
f(x-v_{p+1}) 
\label{eq:Hamiltonian-on-symmetric-space-lem} \\ 
&=\sum_{l=1}^{c}q^{c-l}  
(T_{p+l}^{-1}T_{p+l+1}^{-1} \cdots T_{p+c-1}^{-1}).f(x-v_{p+c}). 
\nonumber 
\end{align} 
\end{lem}

\begin{proof}
For $p+1\le j<m \le p+c$, it holds that  
$T_{j}.f(x-v_{j})=f(x-v_{j+1})$ and  
$T_{m}^{-1}.f(x-v_{j})=f(x-v_{j})$ 
because $a_{j}(x-v_{j})=-1<0$ and $a_{m}(x-v_{j})=0$. 
Using $qT_{j}^{-1}=T_{j}-(1-q)$ repeatedly, we see that 
\begin{align*}
q^{l-1}(T_{p+l-1}^{-1} \cdots T_{p+1}^{-1}).f(x-v_{p+1})=f(x-v_{p+l})-(1-q)\sum_{j=1}^{l-1}q^{l-j-1}f(x-v_{p+j}).  
\end{align*} 
Hence the left hand side of \eqref{eq:Hamiltonian-on-symmetric-space-lem} is equal to
$\sum_{l=1}^{c}q^{c-l}f(x-v_{p+l})$. 
Since $a_{j}(x-v_{j+1})=1\ge 0$ for $p+1\le j <p+c$, we have 
$f(x-v_{p+l})=(T_{p+l}^{-1} \cdots T_{p+c-1}^{-1}).f(x-v_{p+c})$ for $1\le l \le c$.  
This completes the proof. 
\end{proof}

\begin{proof}[Proof of Proposition \ref{prop:Hamiltonian-on-symmetric-space}]
Note that $T_{i}.f(x)=f(x)$ if $a_{i}(x)=0$. 
Using \eqref{eq:main} and \eqref{eq:symmetric-value}, we see that 
\begin{align*}
\sum_{i=1}^{k}q^{d_{i}^{-}(x)}T_{i}^{(-)}(x)T_{i}^{(+)}(x).f(x)=
\sum_{i=1}^{k} q^{d_{i}^{-}(x)}d_{i}^{+}(x)f(x)=\frac{1}{1-q}\sum_{a=1}^{r}(c_{a}-[c_{a}]_{q})f(x). 
\end{align*} 
Hence it suffices to show that 
\begin{align}
& 
\sum_{i=1}^{k}q^{d_{i}^{-}(x)}T_{i}^{(-)}(x).f(x-v_{i}) 
\label{eq:Hamiltonian-on-symmetric-space} \\ 
&=T_{w_{x}}^{-1}\sum_{a=1}^{r}\sum_{l=1}^{c_{a}}q^{c_{a}-l}
(T_{c_{1}+\cdots +c_{a-1}+l-1}^{-1} \cdots T_{c_{1}+\cdots +c_{a}-1}^{-1}).
f(w_{x}x-v_{c_{1}+\cdots +c_{a}}).  
\nonumber 
\end{align}

Fix $1\le i \le k$ and set $y=x-v_{i}$. 
{}From \eqref{eq:main-0} and \eqref{eq:symmetric-value} we have 
\begin{align*}
f(y)=T_{w_{x}}^{-1}(T_{\sigma_{x}(i)}^{-1} \cdots T_{\sigma_{x}(i)+d_{i}^{+}(x)-1}^{-1})f(w_{y}y). 
\end{align*}
Therefore
\begin{align*}
T_{i}^{(-)}(x).f(y)=
T_{w_{x}}^{-1}(T_{\sigma_{x}(i)-1}^{-1} \cdots T_{\sigma_{x}(i)-d_{i}^{-}(x)}^{-1})
(T_{\sigma_{x}(i)-d_{i}^{-}(x)}^{-1} \cdots T_{\sigma_{x}(i)+d_{i}^{+}(x)-1}^{-1}).f(w_{y}y).  
\end{align*}
Since $w_{y}y=w_{x}x-v_{\sigma_{x}(i)+d_{i}^{+}(x)}$ and 
$a_{j}(w_{x}x)=0$ for $\sigma_{x}(i)-d_{i}^{-}(x) \le j<\sigma_{x}(i)+d_{i}^{+}(x)$, 
the right hand side above is equal to 
\begin{align*}
T_{w_{x}}^{-1}(T_{\sigma_{x}(i)-1}^{-1} \cdots T_{\sigma_{x}(i)-d_{i}^{-}(x)}^{-1}).
f(w_{x}x-v_{\sigma_{x}(i)-d_{i}^{-}(x)}).  
\end{align*}
Note that if $c_{1}+\cdots +c_{a-1}<\sigma_{x}(i)\le c_{1}+\cdots +c_{a}$, 
then $\sigma_{x}(i)-d_{i}^{-}(x)=c_{1}+\cdots +c_{a-1}+1$, which is independent of $i$. 
Thus the left hand side of \eqref{eq:Hamiltonian-on-symmetric-space} is equal to 
\begin{align*}
T_{w_{x}}^{-1}\sum_{a=1}^{r}\sum_{l=1}^{c_{a}}q^{l-1}
(T_{c_{1}+\cdots +c_{a-1}+l-1}^{-1} \cdots T_{c_{1}+\cdots +c_{a-1}+1}^{-1}).
f(w_{x}x-v_{c_{1}+\cdots +c_{a-1}+1}). 
\end{align*} 
Now the equality \eqref{eq:Hamiltonian-on-symmetric-space} is an immediate consequence 
of Lemma \ref{lem:Hamiltonian-on-symmetric-space}. 
\end{proof}

Now let us prove Theorem \ref{thm:invariance}. 

\begin{proof}[Proof of Theorem \ref{thm:invariance}]
Suppose that $f \in F_{0}(L, M)$, $x \in L_{+}$ and $1\le i<k$.  
Note that $w_{s_{i}x}s_{i}x=w_{x}x$ and hence  
the cluster coordinates of $x$ and $s_{i}x$ are the same. 

If $a_{i}(x)>0$, then $T_{w_{s_{i}x}}=T_{w_{x}}T_{i}$ from Proposition \ref{prop:shortest-prod}.  
Using \eqref{eq:Hamiltonian-on-symmetric-space-formula}, 
we see that $(Hf)(s_{i}x)=T_{i}^{-1}.(Hf)(x)$.  
Changing $x$ to $s_{i}x$, we find that $(Hf)(s_{i}x)=T_{i}.(Hf)(x)$ if $a_{i}(x)<0$. 

Let us consider the case where $a_{i}(x)=0$. 
{}From Proposition \ref{prop:shortest-prod}, it holds that 
$T_{i}^{-1}T_{w_{x}}^{-1}=T_{w_{x}}^{-1}T_{\sigma_{x}(i)}^{-1}$. 
Note that $d_{\sigma_{x}(i)}^{+}(w_{x}x)=d_{i}^{+}(x)>0$ because $\epsilon_{i+1}(x)=\epsilon_{i}(x)$. 
Hence there exists $1\le a \le r$ such that 
$c_{1}+\cdots +c_{a-1}<\sigma_{x}(i)<c_{1}+\cdots +c_{a}$. 
Now Lemma \ref{lem:invariance} below implies that $T_{i}^{-1}.(Hf)(x)=(Hf)(x)$. 
\end{proof}

\begin{lem}\label{lem:invariance}
Suppose that $x \in L_{+}, f \in F_{0}(L, M)$ and $0\le p<p+c \le k$. 
Set 
\begin{align*}
J=\sum_{l=1}^{c}q^{c-l}(T_{p+l}^{-1}T_{p+l+1}^{-1} \cdots T_{p+c-1}^{-1}).  
\end{align*}
If $a_{j}(x)=0$ for $p+1\le j<p+c$,  
then 
\begin{align*}
T_{p+i}^{-1}\,J.f(x-v_{p+c})=J.f(x-v_{p+c})
\end{align*}
for $1 \le i\le c-1$.  
\end{lem}

\begin{proof}
Using the quadratic relation $q^{-1}T_{j}^{-2}=1-(1-q)T_{j}^{-1}$, we have
\begin{align*}
T_{p+i}^{-1}\,J&=
\sum_{l=1}^{i-1}q^{c-l}(T_{p+l}^{-1} \cdots T_{p+c-1}^{-1})T_{p+i-1}^{-1}+
\sum_{l=i, i+1}q^{c-l}(T_{p+l}^{-1} \cdots T_{p+c-1}^{-1}) \\ 
&+\sum_{l=i+2}^{c}q^{c-l}(T_{p+l}^{-1} \cdots T_{p+c-1}^{-1})T_{p+i}^{-1}. 
\end{align*} 
Note that the first and the third term in the right hand side vanish if $i=1$ and $i=c-1$, respectively. 
Since $T_{j}^{-1}.f(x-v_{p+c})=f(x-v_{p+c})$ for $p+1\le j \le p+c-2$, 
we obtain the desired formula. 
\end{proof}

\subsection{Bethe wave functions}

Let $M$ be a left $\mathcal{H}_{k}$-module. 
We construct eigenfunctions of the restriction $H|_{F_{0}(L, M)}$, 
which we call the {\it Bethe wave functions}. 

Denote by $F(L, M)^{\mathcal{H}_{k}}$ the subspace of $F(L, M)$ 
consisting of the $\rho(\mathcal{H}_{k})$-invariant functions, that is, 
\begin{align*}
F(L, M)^{\mathcal{H}_{k}}:=\left\{
f \in F(L, M) \, | \, \rho(T_{i})f=f \,\, \hbox{for} \,\, 1\le i<k \right\}. 
\end{align*}

\begin{prop}
Let $M$ be a left $\mathcal{H}_{k}$-module. 
It holds that $G(F(L, M)^{\mathcal{H}_{k}}) \subset F_{0}(L, M)$.  
Therefore if $h \in F(L, M)^{\mathcal{H}_{k}}$ is an eigenfunction of $\sum_{i=1}^{k}t_{i}$, 
then $G(h)$ is that of the operator $H|_{F_{0}(L, M)}$ with the same eigenvalue. 
\end{prop}

\begin{proof}
Let $f$ be a function which belongs to $F(L, M)^{\mathcal{H}_{k}}$. 
{}From the definition of the propagation operator, 
we see that $(Gf)(x)=T_{w_{x}}^{-1}.f(w_{x}x)$. 

Suppose that $x \in L$ and $1 \le i<k$. 
If $a_{i}(x)>0$, we have $T_{w_{s_{i}x}}=T_{w_{x}}T_{i}$ by Proposition \ref{prop:shortest-prod}. 
Hence $(Gf)(s_{i}x)=T_{i}^{-1}.(Gf)(x)$ because 
$w_{s_{i}x}s_{i}x=w_{x}x$. 
This also implies that $(Gf)(s_{i}x)=T_{i}.(Gf)(x)$ if $a_{i}(x)<0$. 

Let us consider the case where $a_{i}(x)=0$. 
Then we have $\sigma_{x}(i+1)=\sigma_{x}(i)+1$ and 
$T_{i}^{-1}T_{w_{x}}^{-1}=T_{w_{x}}^{-1}T_{\sigma_{x}(i)}^{-1}$ by 
Proposition \ref{prop:shortest-prod}. 
Now note that 
\begin{align*}
a_{\sigma_{x}(i)}(w_{x}x)=\epsilon_{\sigma_{x}(i)}(w_{x}x)-\epsilon_{\sigma_{x}(i+1)}(w_{x}x)=
\epsilon_{i}(x)-\epsilon_{i+1}(x)=a_{i}(x)=0. 
\end{align*}
Hence we find that $T_{\sigma_{x}(i)}^{-1}.f(w_{x}x)=(\rho(T_{\sigma_{x}(i)}^{-1})f)(w_{x}x)=f(w_{x}x)$ from 
Proposition \ref{prop:action-trivial}. 
Therefore $T_{i}^{-1}.(Gf)(x)=(Gf)(x)$.  
\end{proof}

For $1\le i<k$ and $\lambda \in V^{*}$, we define $Y_{i}(\lambda) \in \mathcal{H}_{k}$ by 
\begin{align*}
Y_{i}(\lambda):=
\frac{(e^{\lambda(v_{i})}-e^{\lambda(v_{i+1})})T_{i}-e^{\lambda(v_{i})}(\alpha\,e^{\lambda(v_{i+1})}+1-q)}
{\alpha\,e^{\lambda(v_{i}+v_{i+1})}+e^{\lambda(v_{i})}-qe^{\lambda(v_{i+1})}}.  
\end{align*}

\begin{lem}\label{lem:comm-rel-Y}
The following equalities hold. 
\begin{enumerate}
 \item $Y_{i}(s_{i}\lambda)Y_{i}(\lambda)=1$ for $1 \le i<k$ and $\lambda \in V^{*}$. 
 \item 
$Y_{i+1}(s_{i}s_{i+1}\lambda)Y_{i}(s_{i+1}\lambda)Y_{i+1}(\lambda)=
Y_{i}(s_{i+1}s_{i}\lambda)Y_{i+1}(s_{i}\lambda)Y_{i}(\lambda)$ for $1\le i \le k-2$ and $\lambda \in V^{*}$. 
\end{enumerate}
\end{lem}

\begin{proof}
By a direct computation.  
\end{proof}

{}From Lemma \ref{lem:comm-rel-Y}, we obtain the following proposition. 

\begin{prop}
Suppose that $\lambda \in V^{*}$. 
There exists a unique $\mathbb{C}$-algebra homomorphism $\phi_{\lambda}: \mathbb{C}[W] \to \mathcal{H}_{k}$ 
such that $\phi_{\lambda}(1)=1$ and 
\begin{align}
\phi_{\lambda}(s_{i}w)=Y_{i}(w\lambda)\phi_{\lambda}(w) 
\label{eq:phi-cocycle}
\end{align}
for $1 \le i<k$ and $w \in W$.  
\end{prop}

\begin{thm}\label{thm:wave-function}
Let $M$ be a left $\mathcal{H}_{k}$-module. 
For $\lambda \in V^{*}$ and $m \in M$, we define a function $h_{\lambda}^{m} \in F(L, M)$ by 
\begin{align}
h_{\lambda}^{m}(x):=\sum_{w \in W}e^{(w\lambda)(x)}(\phi_{\lambda}(w).m) \quad (x \in L),  
\label{eq:wave-function} 
\end{align}
where $.$ means the left action of $\mathcal{H}_{k}$ on $M$. 
Then the function $h_{\lambda}^{m}$ belongs to $F(L, M)^{\mathcal{H}_{k}}$ and 
is an eigenfunction of $\sum_{i=1}^{k}t_{i}$ with eigenvalue $\sum_{i=1}^{k}e^{-\lambda(v_{i})}$. 
Therefore $G(h_{\lambda}^{m})$ is an eigenfunction of $H|_{F_{0}(L, M)}$. 
\end{thm}

\begin{proof}
It is clear that $h_{\lambda}^{m}$ is an eigenfunction of $\sum_{i=1}^{k}t_{i}$. 
{}From the relation \eqref{eq:phi-cocycle}, we have 
\begin{align*}
T_{j}\phi_{\lambda}(w)&=
\frac{e^{w\lambda(v_{j})}(\alpha\, e^{w\lambda(v_{j+1})}+1-q)}{e^{w\lambda(v_{j})}-e^{w\lambda(v_{j+1})}}
\phi_{\lambda}(w) \\ 
&+\frac{\alpha e^{w\lambda(v_{j}+v_{j+1})}+e^{w\lambda(v_{j})}-qe^{w\lambda(v_{j+1})}}
{e^{w\lambda(v_{j})}-e^{w\lambda(v_{j+1})}}\phi_{\lambda}(s_{j}w) 
\end{align*}
for $1 \le j<k$ and $w \in W$. 
Moreover, from the definition of $\widehat{I}_{j}$, we see that 
\begin{align*}
\widehat{I}_{j}(e^{\mu})=
\frac{\alpha \, e^{\mu(v_{j}+v_{j+1})}+(1-q)e^{\mu(v_{j})}}
{e^{\mu(v_{j})}-e^{\mu(v_{j+1})}}
\left(e^{\mu}-e^{s_{j}\mu}\right) \quad (\mu \in V^{*}). 
\end{align*}
Combining the above equalities, we find that $\rho(T_{i})h_{\lambda}^{m}=h_{\lambda}^{m}$ for $1 \le i<k$. 
\end{proof}

Let $M$ be a left $\mathcal{H}_{k}$-module. 
Set 
\begin{align*}
\mathcal{F}(L_{+}, M):=\{ f : L_{+} \to M  \, | \, 
T_{i}.f(x)=f(x) \,\, \hbox{if} \,\, a_{i}(x)=0 \}.  
\end{align*}
We identify $F_{0}(L, M)$ with $\mathcal{F}(L_{+}, M)$ by
the map $f \mapsto f|_{L+} \, (f \in F_{0}(L, M))$. 
Denote by $H^{+}$ the restriction of the discrete Hamiltonian $H$ to $\mathcal{F}(L_{+}, M)$. 
Proposition \ref{prop:Hamiltonian-on-symmetric-space} implies that  
the operator $H^{+}$ is given by 
\begin{align}
(H^{+}f)(x)&=\sum_{a=1}^{r}\sum_{l=1}^{c_{a}}
q^{c_{a}-l}\,T_{c_{1}+\cdots +c_{a-1}+l}^{-1} T_{c_{1}+\cdots +c_{a-1}+l+1}^{-1} \cdots 
T_{c_{1}+\cdots +c_{a}-1}^{-1}.f(x-v_{c_{1}+\cdots +c_{a}}) 
\label{eq:Hamiltonian-restricted}
\\
& \quad{}-\frac{\alpha}{1-q}\sum_{a=1}^{r}(c_{a}-[c_{a}]_{q})f(x) \quad 
(f \in \mathcal{F}(L_{+}, M),  x \in L_{+}), 
\nonumber 
\end{align} 
where $(c_{1}, c_{2}, \ldots , c_{r})$ is the cluster coordinate of $x$.

\begin{cor}
Let $M$ be a left $\mathcal{H}_{k}$-module. 
Suppose that $\lambda \in V^{*}$ and $m \in M$, 
and consider the function $h_{\lambda}^{m}$ defined by \eqref{eq:wave-function}.  
Then $h_{\lambda}^{m}|_{L_{+}}$ belongs to $\mathcal{F}(L_{+}, M)$ and is an eigenfunction of $H^{+}$ 
with eigenvalue $\sum_{i=1}^{k}e^{-\lambda(v_{i})}$.  
\end{cor}

\begin{proof}
It follows from Theorem \ref{thm:wave-function} and the equality 
$G(h_{\lambda}^{m})|_{L_{+}}=h_{\lambda}^{m}|_{L_{+}}$.  
\end{proof}


\section{Algebraic construction of multi-species $q$-Boson system}\label{sec:system}

\subsection{Setting}

In the rest of this article we fix a positive integer $N$. 
For a positive integer $c$, set 
\begin{align*}
I_{N, c}&:=\{1, 2, \ldots , N\}^{c}, \\ 
I_{N, c}^{+}&:=\{ 
(\mu_{1}, \ldots , \mu_{c}) \in I_{N, c} \, | \, \mu_{1}\le \cdots \le \mu_{c}\}. 
\end{align*}

Let $x \in L_{+}$ and $(c_{1}, \ldots , c_{r})$ be the cluster coordinate of $x$. 
For $\boldsymbol{\mu} \in I_{N, k}$, 
we define $\boldsymbol{\mu}[x] \in I_{N, k}$ 
as follows. 
According to the decomposition 
$I_{N, k}=I_{N, c_{1}}\times \cdots \times I_{N, c_{r}}$, 
we write $\boldsymbol{\mu}=(\boldsymbol{\mu}_{1}, \ldots , \boldsymbol{\mu}_{r})$, 
where $\boldsymbol{\mu}_{a} \in I_{N, c_{a}} \, (1\le a \le r)$. 
Then let $\boldsymbol{\mu}_{a}^{+}$ be  
the unique element of $I_{N, c_{a}}^{+}$ which is a rearrangement of $\boldsymbol{\mu}_{a}$.  
Now set 
$\boldsymbol{\mu}[x]:=(\boldsymbol{\mu}_{1}^{+}, \ldots , \boldsymbol{\mu}_{r}^{+})$. 
For example, if $k=5$, $x=2v_{1}+2v_{2}-v_{3}-v_{4}-v_{5}$ and $\boldsymbol{\mu}=(3, 1, 4, 2, 5)$, 
then $(c_{1}, c_{2})=(2, 3), \boldsymbol{\mu}_{1}=(3, 1), \boldsymbol{\mu}_{2}=(4, 2, 5)$, and  
$\boldsymbol{\mu}[x]=(1, 3, 2, 4, 5)$.

Set 
\begin{align*}
\mathcal{S}:=\left\{ 
(x, \boldsymbol{\nu}) \in L_{+}\times I_{N, k} \, | \, \boldsymbol{\nu}=\boldsymbol{\nu}[x]
\right\}. 
\end{align*}
We identify $\mathcal{S}$ with the set of configurations of $k$ bosonic particles 
of $N$ species on the one-dimensional lattice $\mathbb{Z}$ as follows. 
For $x=\sum_{i=1}^{k}m_{i}v_{i} \in L_{+}$ and 
$\boldsymbol{\nu}=(\nu_{1}, \ldots , \nu_{k})$, 
we assign to $(x, \boldsymbol{\nu})$ the configuration 
such that the particles with the color $\nu_{1}, \ldots , \nu_{k}$ 
are on the sites $m_{1}, \ldots , m_{k}$, respectively.  
For example, if $k=6, N=4, x=2v_{1}+2v_{2}-v_{3}-3v_{4}-3v_{5}-3v_{6}$ 
and $\boldsymbol{\nu}=(1, 2, 4, 2, 2, 3)$, 
then $(x, \boldsymbol{\nu})$ corresponds to the configuration in Figure \ref{fig2}. 

Denote the set of $\mathbb{R}$-valued functions on $\mathcal{S}$ by $F(\mathcal{S})$. 
In the rest of this paper 
we construct the transition rate matrix $Q: F(\mathcal{S}) \to F(\mathcal{S})$ of a 
continuous-time Markov process on $\mathcal{S}$. 

\begin{figure}
\centering
{\unitlength 0.1in%
\begin{picture}(32.0000,12.8500)(4.0000,-19.3500)%
%
\special{pn 13}%
\special{pa 400 1800}%
\special{pa 3600 1800}%
\special{fp}%
\special{sh 1}%
\special{pa 3600 1800}%
\special{pa 3533 1780}%
\special{pa 3547 1800}%
\special{pa 3533 1820}%
\special{pa 3600 1800}%
\special{fp}%
%
\special{pn 13}%
\special{ar 3000 1600 150 150 0.0000000 6.2831853}%
%
\special{pn 13}%
\special{ar 3000 1600 150 150 0.0000000 6.2831853}%
%
\special{pn 13}%
\special{ar 3000 1200 150 150 0.0000000 6.2831853}%
%
\special{pn 13}%
\special{ar 1800 1600 150 150 0.0000000 6.2831853}%
%
\special{pn 13}%
\special{ar 1000 1600 150 150 0.0000000 6.2831853}%
%
\special{pn 13}%
\special{ar 1000 1200 150 150 0.0000000 6.2831853}%
%
\special{pn 13}%
\special{ar 1000 800 150 150 0.0000000 6.2831853}%
\put(10.0000,-8.0000){\makebox(0,0){$2$}}%
\put(10.0000,-12.0000){\makebox(0,0){$2$}}%
\put(10.0000,-16.0000){\makebox(0,0){$3$}}%
\put(18.0000,-16.0000){\makebox(0,0){$4$}}%
\put(30.0000,-12.0000){\makebox(0,0){$1$}}%
\put(30.0000,-16.0000){\makebox(0,0){$2$}}%
\put(30.0000,-20.0000){\makebox(0,0){$2$}}%
\put(26.0000,-20.0000){\makebox(0,0){$1$}}%
\put(22.0000,-20.0000){\makebox(0,0){$0$}}%
\put(18.0000,-20.0000){\makebox(0,0){$-1$}}%
\put(14.0000,-20.0000){\makebox(0,0){$-2$}}%
\put(10.0000,-20.0000){\makebox(0,0){$-3$}}%
\end{picture}}%
\caption{}
\label{fig2}
\end{figure}

\subsection{Derivation of the transition rate matrix}

Hereafter we assume that 
\begin{align*}
0<q<1.  
\end{align*}
Let $U=\oplus_{\mu=1}^{N}\mathbb{C}u_{\mu}$ be the $N$-dimensional vector space 
with the basis $\{u_{\mu}\}_{\mu=1}^{N}$.  
Consider the $\mathbb{C}$-linear operator $R \in \mathrm{End}(U^{\otimes 2})$ defined by 
\begin{align*}
R(u_{\mu} \otimes u_{\mu'})=\left\{ 
\begin{array}{ll}
q^{1/2}u_{\mu'}\otimes u_{\mu} & (\mu>\mu'), \\
u_{\mu} \otimes u_{\mu} & (\mu=\mu'), \\ 
(1-q)u_{\mu}\otimes u_{\mu'}+q^{1/2}u_{\mu'}\otimes u_{\mu} & (\mu<\mu'). 
\end{array}
\right.
\end{align*}
Let $R_{i} \in \mathrm{End}(U^{\otimes k}) \, (1\le i<k) $ be the linear operator 
acting as $R$ on the tensor product of the $i$-th and $(i+1)$-th component of $U^{\otimes k}$.  

\begin{thm}\cite{J}
The assignment $T_{i} \mapsto R_{i} \, (1\le i<k)$ uniquely extends to a $\mathbb{C}$-algebra 
homomorphism $\mathcal{H}_{k} \to \mathrm{End}(U^{\otimes k})$.  
\end{thm}

In the following we regard $U^{\otimes k}$ as a left $\mathcal{H}_{k}$-module 
with respect to the action defined above. 

For $\boldsymbol{\mu}=(\mu_{1}, \ldots , \mu_{k}) \in I_{N, k}$, we set 
\begin{align*}
\boldsymbol{u}_{\boldsymbol{\mu}}:=u_{\mu_{1}} \otimes \cdots \otimes u_{\mu_{k}} \in U^{\otimes k}, 
\end{align*}
and 
\begin{align*}
\ell(\boldsymbol{\mu}):=\#
\{ (i, j) \, | \, 1\le i<j \le k \,\, \hbox{and} \,\, \mu_{i}>\mu_{j} \}. 
\end{align*}
Then any function $f: L_{+} \to U^{\otimes k}$ is uniquely written in the form 
\begin{align}
f(x)=\sum_{\boldsymbol{\mu}\in I_{N, k}}q^{\ell(\boldsymbol{\mu})/2}
f_{\boldsymbol{\mu}}(x)\,\boldsymbol{u}_{\boldsymbol{\mu}} \quad (x \in L_{+}), 
\label{eq:symmetric-function-decompose} 
\end{align}
where $f_{\boldsymbol{\mu}}$ is a $\mathbb{C}$-valued function on $L_{+}$. 
Then the space $\mathcal{F}(L_{+}, U^{\otimes k})$ has the following description.

\begin{prop}\label{prop:symmetric-function-decompose}
In the above notation the following statements are equivalent. 
\begin{enumerate}
 \item $f \in \mathcal{F}(L_{+}, U^{\otimes k})$. 
 \item Suppose that $x \in L_{+}$ and $1\le i<k$. If $a_{i}(x)=0$ then 
$f_{\ldots , \mu_{i}, \mu_{i+1}, \ldots}(x)=f_{\ldots , \mu_{i+1}, \mu_{i}, \ldots}(x)$ 
for all $\boldsymbol{\mu}=(\mu_{1}, \ldots , \mu_{k}) \in I_{N, k}$. 
\end{enumerate} 
\end{prop}

\begin{proof}
It follows from the definition of the operator $R$ and 
\eqref{eq:symmetric-function-decompose} by a direct computation. 
\end{proof}

We define the $\mathbb{C}$-linear map $\varphi: F(\mathcal{S}) \to \mathcal{F}(L_{+}, U^{\otimes k})$ by 
\begin{align*}
(\varphi h)(x):=\sum_{\boldsymbol{\mu}\in I_{N, k}} q^{\ell(\boldsymbol{\mu})/2}h(x, \boldsymbol{\mu}[x])\,
\boldsymbol{u}_{\boldsymbol{\mu}} \quad 
(h \in F(\mathcal{S}), x \in L_{+}). 
\end{align*}
Proposition \ref{prop:symmetric-function-decompose} implies that  
the function defined by the right hand side above belongs to $\mathcal{F}(L_{+}, U^{\otimes k})$. 
The map $\varphi$ is an isomorphism with the inverse 
\begin{align*}
(\varphi^{-1}f)(x, \boldsymbol{\nu})=f_{\boldsymbol{\nu}}(x) \quad 
(f \in \mathcal{F}(L_{+}, U^{\otimes k}), \, (x, \boldsymbol{\nu}) \in \mathcal{S}), 
\end{align*}
where $f_{\boldsymbol{\nu}}$ is defined by \eqref{eq:symmetric-function-decompose}.

Now let us write down the operator $\varphi^{-1}H^{+}\varphi \, : \, F(\mathcal{S}) \to F(\mathcal{S})$. 
For $1\le c \le k$, we set 
\begin{align*}
A^{(c)}:=\sum_{l=1}^{c}q^{c-l}T_{l}^{-1}T_{l+1}^{-1}\cdots T_{c-1}^{-1} \in \mathcal{H}_{k}.  
\end{align*}
It acts on $U^{\otimes c}$. 
We set the matrix element 
$A_{\boldsymbol{\mu},\boldsymbol{\mu}'}^{(c)} \in \mathbb{R}$ by 
\begin{align*}
A^{(c)}\boldsymbol{u}_{\boldsymbol{\mu}}=
\sum_{\boldsymbol{\mu}' \in I_{N, c}}A_{\boldsymbol{\mu},\boldsymbol{\mu}'}^{(c)}
\boldsymbol{u}_{\boldsymbol{\mu}'}. 
\end{align*}

Suppose that $h \in F(\mathcal{S})$ and $(x, \boldsymbol{\nu}) \in \mathcal{S}$.   
Denote by $(c_{1}, \ldots , c_{r})$ the cluster coordinate of $x$, and 
decompose $\boldsymbol{\nu}=(\boldsymbol{\nu}_{1}, \ldots , \boldsymbol{\nu}_{r})$, 
where $\boldsymbol{\nu}_{a} \in I_{N, c_{a}}^{+}$ for $1\le a\le r$. 
{}From \eqref{eq:Hamiltonian-restricted} we see that 
\begin{align}
(\varphi^{-1}H^{+}\varphi h)(x, \boldsymbol{\nu})&=
\sum_{a=1}^{r}\sum_{\boldsymbol{\eta} \in I_{N, c_{a}}}
q^{\ell(\boldsymbol{\eta})/2}
A_{\boldsymbol{\eta}, \boldsymbol{\nu}_{a}}^{(c_{a})}\,
(\varphi h)_{(\boldsymbol{\nu}_{1}, \ldots , \boldsymbol{\nu}_{a-1}, 
\boldsymbol{\eta}, \boldsymbol{\nu}_{a+1}, \ldots , \boldsymbol{\nu}_{r})}
(x-v_{c_{1}+\cdots +c_{a}})
\label{eq:H-conjugate} \\ 
&-\frac{\alpha}{1-q}\sum_{a=1}^{c_{a}}(c_{a}-[c_{a}]_{q})
h(x, \boldsymbol{\nu}). 
\nonumber 
\end{align}

\begin{prop}\label{prop:matrix-element-A}
Suppose that $1 \le c\le k$ and 
\begin{align}
\boldsymbol{\nu}=(\underbrace{1, \ldots , 1}_{m_{1}}, \ldots , \underbrace{N, \ldots , N}_{m_{N}}) 
\in I_{N, c}^{+},   
\label{eq:matrix-element-A-0}
\end{align}
where $m_{1}, \ldots , m_{N}$ are non-negative integers satisfying $\sum_{i=1}^{N}m_{i}=c$. 
Then the matrix element $A_{\boldsymbol{\eta}, \boldsymbol{\nu}}^{(c)}$ is zero unless 
$\boldsymbol{\eta}$ is of the form 
\begin{align}
(\underbrace{1, \ldots , 1}_{m_{1}}, \ldots , 
\underbrace{b, \ldots , b}_{m_{b}-1}, \ldots 
\underbrace{N, \ldots , N}_{m_{N}}, b)  
\label{eq:matrix-element-A}
\end{align} 
for some $1\le b \le N$. 
If $\boldsymbol{\eta}$ is equal to \eqref{eq:matrix-element-A}, then 
\begin{align*}
A_{\boldsymbol{\eta}, \boldsymbol{\nu}}^{(c)}=
\frac{1-q^{m_{b}}}{1-q}q^{\sum_{i=b+1}^{N}m_{i}/2}=
\frac{1-q^{m_{b}}}{1-q}q^{\ell(\boldsymbol{\eta})/2}.  
\end{align*}
\end{prop}

\begin{proof}
For $n \ge 1$, 
let $( \, , \, )$ be the non-degenerate bilinear form on $U^{\otimes n}$ defined by 
\begin{align*}
(u_{\mu_{1}}\otimes \cdots \otimes u_{\mu_{n}}, u_{\nu_{1}}\otimes \cdots \otimes u_{\nu_{n}})=
\delta_{\mu_{1}\nu_{1}} \cdots \delta_{\mu_{n}\nu_{n}}.  
\end{align*}
Consider the linear operator $S \in \mathrm{End}(U^{\otimes 2})$ defined by 
\begin{align*}
S(u_{\mu} \otimes u_{\mu'})=\left\{ 
\begin{array}{ll}
(1-q^{-1})u_{\mu}\otimes u_{\mu'}+q^{-1/2}u_{\mu'}\otimes u_{\mu} & (\mu>\mu'), \\
u_{\mu} \otimes u_{\mu} & (\mu=\mu'), \\ 
q^{-1/2}u_{\mu'}\otimes u_{\mu} & (\mu<\mu'). 
\end{array}
\right.
\end{align*} 
Then it holds that $(R^{-1}\boldsymbol{u}, \boldsymbol{u'})=(\boldsymbol{u}, S\boldsymbol{u'})$ 
for $\boldsymbol{u}, \boldsymbol{u'} \in U^{\otimes 2}$. 
Let $S_{i} \in \mathrm{End}(U^{\otimes k}) \, (1 \le i<k)$ be the linear operator acting as $S$ 
on the tensor product of the $i$-th and $(i+1)$-th component of $U^{\otimes k}$. 
Then we have 
\begin{align*}
A_{\boldsymbol{\eta}, \boldsymbol{\nu}}^{(c)}=
(A^{(c)}\boldsymbol{u}_{\boldsymbol{\eta}}, \boldsymbol{u}_{\boldsymbol{\nu}})=
\sum_{l=1}^{c}q^{c-l}(\boldsymbol{u}_{\boldsymbol{\eta}}, S_{c-1}\cdots S_{l}\boldsymbol{u}_{\boldsymbol{\nu}}).  
\end{align*} 
{}From the definition of $S$, we see that 
\begin{align*}
S_{c-1}\cdots S_{l}\boldsymbol{u}_{\boldsymbol{\nu}}=
q^{-\sum_{i=b+1}^{N}m_{i}/2}\boldsymbol{u}_{\boldsymbol{\nu^{(b)}}},  
\end{align*}
where $b$ is determined by the condition $m_{1}+\cdots +m_{b-1}<l\le m_{1}+\cdots +m_{b}$ 
and $\boldsymbol{\nu^{(b)}}$ is the tuple \eqref{eq:matrix-element-A}. 
Therefore $A_{\boldsymbol{\eta}, \boldsymbol{\nu}}^{(c)}=0$ unless 
$\boldsymbol{\eta}$ is of the form \eqref{eq:matrix-element-A}. 
If $\boldsymbol{\eta}$ is equal to \eqref{eq:matrix-element-A}, it holds that 
\begin{align*}
A_{\boldsymbol{\eta}, \boldsymbol{\nu}}^{(c)}&=
\sum_{l=m_{1}+\cdots +m_{b-1}+1}^{m_{1}+\cdots +m_{b}}q^{c-l-\sum_{i=b+1}^{N}m_{i}/2} \\ 
&=q^{\sum_{i=b+1}^{N}m_{i}/2}\sum_{l=m_{1}+\cdots +m_{b-1}+1}^{m_{1}+\cdots +m_{b}}q^{\sum_{i=1}^{b}m_{i}-l}=
q^{\sum_{i=b+1}^{N}m_{i}/2}\frac{1-q^{m_{b}}}{1-q}.  
\end{align*}
Here we used $c=\sum_{i=1}^{N}m_{i}$. 
This completes the proof. 
\end{proof}

For $\boldsymbol{\nu} \in I_{N, c}^{+}$ of the form \eqref{eq:matrix-element-A-0} 
and $1\le b \le N$, 
we set 
\begin{align*}
\boldsymbol{\nu}^{b, \pm}:=
(\underbrace{1, \ldots , 1}_{m_{1}}, \ldots , 
\underbrace{b, \ldots , b}_{m_{b}\pm 1}, \ldots 
\underbrace{N, \ldots , N}_{m_{N}}) \in I_{N, c\pm 1}^{+}.   
\end{align*}
Let $(x, \boldsymbol{\nu}) \in \mathcal{S}$.  
Denote by $(c_{1}, \ldots , c_{r})$ the cluster coordinate of $x$  
and write $\boldsymbol{\nu}=(\boldsymbol{\nu}_{1}, \ldots , \boldsymbol{\nu}_{r})$  
where $\boldsymbol{\nu}_{a} \in I_{N, c_{a}}^{+} \, (1\le a \le r)$. 
For $(y, \boldsymbol{\eta}) \in \mathcal{S}$, 
we write $(x, \boldsymbol{\nu}) \rightsquigarrow (y, \boldsymbol{\eta})$ if the following conditions hold: 
\begin{enumerate}
 \item 
 $y=x-v_{c_{1}+\cdots +c_{a}}$ for some $1\le a \le r$. 
 \item 
If $\epsilon_{c_{1}+\cdots +c_{a}}(x)-1>\epsilon_{c_{1}+\cdots +c_{a}+1}(x)$, then 
\begin{align*}
\boldsymbol{\eta}=(\boldsymbol{\nu}_{1}, \ldots , \boldsymbol{\nu}_{a-1}, 
\boldsymbol{\nu}_{a}^{b, -}, b, \boldsymbol{\nu}_{a+1}, \ldots , \boldsymbol{\nu}_{r}) 
\end{align*}
for some $1\le b \le N$. 
If $\epsilon_{c_{1}+\cdots +c_{a}}(x)-1=\epsilon_{c_{1}+\cdots +c_{a}+1}(x)$, then 
\begin{align*}
\boldsymbol{\eta}=(\boldsymbol{\nu}_{1}, \ldots , \boldsymbol{\nu}_{a-1}, 
\boldsymbol{\nu}_{a}^{b, -}, \boldsymbol{\nu}_{a+1}^{b, +}, \boldsymbol{\nu}_{a+2}, \ldots , \boldsymbol{\nu}_{r}) 
\end{align*}
for some $1\le b \le N$. 
\end{enumerate}
Moreover, when the above conditions are satisfied, we set 
\begin{align}
c(x, \boldsymbol{\nu} | y, \boldsymbol{\eta}):=\frac{1-q^{m_{b}}}{1-q} q^{\sum_{i=b+1}^{N}m_{i}},  
\label{eq:transition-rate}
\end{align}
where $m_{i} \, (1\le i \le N)$ is the number of $i$ in $\boldsymbol{\nu}_{a}$. 
We set $c(x, \boldsymbol{\nu} | y, \boldsymbol{\eta})=0$ unless 
$(x, \boldsymbol{\nu})\rightsquigarrow (y, \boldsymbol{\eta})$. 

Under the identification of $\mathcal{S}$ with the set of configurations of
$k$ bosonic particles of $N$ species given in the previous subsection,  
the condition $(x, \boldsymbol{\nu}) \rightsquigarrow (y, \boldsymbol{\eta})$ means that 
the configuration corresponding to $(y, \boldsymbol{\eta})$ is obtained from 
that corresponding to $(x, \boldsymbol{\nu})$ by moving one particle to the left. 

{}From \eqref{eq:H-conjugate} and the equality 
\begin{align*}
\sum_{b=1}^{N}\frac{1-q^{m_{b}}}{1-q} q^{\sum_{i=b+1}^{N}m_{i}}=\frac{1-q^{\sum_{i=1}^{N}m_{i}}}{1-q}, 
\end{align*}
we see that 
\begin{align*}
(\varphi^{-1}H^{+}\varphi h)(x, \boldsymbol{\nu})&=
\sum_{(y, \boldsymbol{\eta}) \in \mathcal{S}}c(x, \boldsymbol{\nu} | y, \boldsymbol{\eta})
\left\{h(y, \boldsymbol{\eta})-h(x, \boldsymbol{\nu})\right\} \\ 
& {}+\left(\frac{1-q+\alpha}{1-q}\sum_{a=1}^{r}[c_{a}]_{q}-\frac{\alpha}{1-q}k\right)h(x, \boldsymbol{\nu}) 
\quad (h \in F(\mathcal{S}), (x, \boldsymbol{\nu}) \in \mathcal{S}), 
\end{align*}
where $(c_{1}, \ldots , c_{r})$ is the cluster coordinate of $x$. 
{}From the above formula we obtain the following result. 

\begin{thm}
Set $\alpha=-(1-q)$ and suppose that $0<q<1$. 
Then the operator 
\begin{align*}
Q:=\varphi^{-1}H^{+}\varphi-k 
\end{align*} 
gives the transition rate matrix of a continuous-time Markov process on $\mathcal{S}$. 
\end{thm}

The resulting process is described as follows. 
In continuous time one particle may move from site $i$ to $i-1$ 
independently for each $i \in \mathbb{Z}$. 
The transition rate at which a particle with color $b$ moves 
is given by the right hand side of \eqref{eq:transition-rate}, 
where $m_{i} \, (1\le i\le N)$ is the number of particles 
with color $i$ in the cluster from which the moving particle leaves.


\section*{Acknowledgments}

The research of the author is supported by 
JSPS KAKENHI Grant Number 26400106.

\end{document}